\documentclass[letterpaper, 10 pt, conference]{ieeeconf}
\IEEEoverridecommandlockouts
\let\labelindent\relax
\usepackage{enumitem}
\usepackage{balance}
\usepackage[font={small}]{caption}
\usepackage{subcaption}
\usepackage{array}
\usepackage{textcomp}
\usepackage{mathtools, nccmath}
\usepackage{graphicx}
\usepackage{amsfonts}
\usepackage{amsmath}
\usepackage{amssymb}
\usepackage{algorithm}
\usepackage{algorithmic}
\usepackage{hyperref}
\usepackage{tikz}
\usepackage{arydshln}
\usepackage{multirow}
\usepackage{bm}
\usepackage{epstopdf}
\usepackage{cite}
\usepackage{siunitx}

\DeclareMathOperator*{\argmin}{argmin} % no space, limits underneath in displays
\usepackage{amsthm}

\newtheoremstyle{mystyle}%                % Name
  {}%                                     % Space above
  {}%                                     % Space below
  {}%                                     % Body font
  {}%                                     % Indent amount
  {\bfseries}%                            % Theorem head font
  {.}%                                    % Punctuation after theorem head
  { }%                                    % Space after theorem head, ' ', or \newline
  {\thmname{#1}\thmnumber{ #2}\thmnote{ (#3)}}%                                     % Theorem head spec (can be left empty, meaning `normal')

\theoremstyle{mystyle}

\newtheorem{theorem}{Theorem}

\newtheorem{remark}{Remark}

\begin{document}
\title{\LARGE \bf Safety-Critical Control using Optimal-decay Control Barrier Function with Guaranteed Point-wise Feasibility}
\author{Jun Zeng, Bike Zhang, Zhongyu Li and Koushil Sreenath
    \thanks{Jun Zeng, Bike Zhang, Zhongyu Li and Koushil Sreenath are with the Department of Mechanical Engineering, University of California, Berkeley, CA, 94720, USA, \tt\small \{zengjunsjtu, bikezhang, zhongyu\_li, koushils\}@berkeley.edu}
    \thanks{This work was partially supported through National Science Foundation Grant CMMI-1931853.}
    \thanks{Code and simulations are available at \url{https://github.com/HybridRobotics/CBF-Pointwise-Feasibility}.}
}
\maketitle

\begin{abstract}
Safety is one of the fundamental problems in robotics.
Recently, a quadratic program based control barrier function (CBF) method has emerged as a way to enforce safety-critical constraints.
Together with control Lyapunov function (CLF), it forms a safety-critical control strategy, named CLF-CBF-QP, which can mediate between achieving the control objective and ensuring safety, while being executable in real-time.
However, once additional constraints such as input constraints are introduced, the CLF-CBF-QP may encounter infeasibility.
In order to address the challenge arises due to the infeasibility, we propose an optimal-decay form for safety-critical control wherein the decay rate of the CBF is optimized point-wise in time so as to guarantee point-wise feasibility when the state lies inside the safe set.
The proposed control design is numerically validated using an adaptive cruise control example.

\end{abstract}

\section{Introduction}
\label{sec:introduction}
\subsection{Motivation}
Safety-critical optimal control and planning is one of the fundamental problems in robotics, \textit{e.g.}, robots need to be able to safely avoid obstacles while using minimal energy.
In order to ensure the safety of robotic systems while achieving optimal performance, the tight coupling between potentially conflicting control objectives and safety criteria is usually formulated as an optimization problem.
Recent work \cite{ames2014control, ames2016control} formulates this problem using control barrier functions as an optimization problem, where the safety criteria is formulated as constraints.
Additionally, a robotic system usually has a constrained set for the admissible inputs due to physical limitations, which can be taken as additional constraints in the optimization problem \cite{nguyen2018dynamic, agrawal2017discrete, choi2020reinforcement}.
However, the input constraint and control barrier function constraint might be in conflict  in the optimization and make the optimization infeasible.
In this paper, we study this feasibility problem and illustrate a variant form of safety-critical control with control barrier functions which satisfies the input constraint and guarantees point-wise feasibility along the trajectory.

\subsection{Related Work}
CBFs have recently been introduced as a promising way to ensure set invariance by considering the system dynamics. %and implying forward invariance of the safe set.
Furthermore, a safety-critical control design for continuous-time systems was proposed by unifying a control Lyapunov function (CLF) and a control barrier function (CBF) through a quadratic program (CLF-CBF-QP) \cite{ames2016control, ames2019control}.
This method could be deployed as a real-time optimization-based controller with safety-critical constraints, shown in \cite{ames2014control, wu2015safety}.
The adaptive, robust, and stochastic cases of safety-critical control with CBF have been considered in \cite{nguyen2017robust, xu2015robustness, taylor2020adaptive}. CBFs have also been used for high relative degree safety constraints for nonlinear systems \cite{nguyen2016exponential, xiao2019control}.
Besides the continuous-time domain, the formulation of CBF was generalized into discrete-time systems in \cite{agrawal2017discrete}, and systems evolving on manifolds in \cite{wu2015safety}. Recently, CBF constraints were also applied in control problems using a data-driven approach \cite{li2019temporal, cheng2019end, taylor2020learning, robey2020learning, lindemann2020learning, robey2021learning} and optimal control design \cite{zeng2020safety, cohen2020approximate}.

However, the input constraint was not considered in the optimization in some of the previous work, such as \cite{ames2019control}, which means the optimized control input might not be executable due to the physical limitations of the system. Other work, such as \cite{wu2015safety, xu2015robustness, nguyen2017robust, nguyen2018dynamic, agrawal2017discrete, choi2020reinforcement}, did consider the input constraint, but the potential conflict between input constraint and CBF constraint was not addressed, potentially resulting in an unsolvable optimization problem.
In \cite{ames2016control}, a valid CBF is specifically designed for adaptive cruise control scenario with explicit integration problem over the system dynamics under input constraint, which could ensure the feasibility in the optimization.
However, the feasibility problem for general nonlinear systems remains as an unsolved topic for the safety-critical optimal control.
\subsection{Contribution}
The contributions of this paper are as follows.
\begin{itemize}
    \item We present the reasons of potential infeasibility in the CBF-QP and CLF-CBF-QP.
    \item With analysis in both input-space and state-space, we reveal quantitatively and qualitatively that the infeasibility appears potentially due to a small decay rate of the control barrier function constraint.
    \item We propose an optimal-decay form of CBF-QP and CLF-CBF-QP where the decay rate of the control barrier function is optimized.
    \item We prove that our optimal-decay formulation is point-wise feasible for any state lying 
    strictly inside the safe set. % THIS FINE?
    % inside a subset of the safe set,
    An adaptive cruise control example is used to numerically verify this.
\end{itemize}

\subsection{Paper Structure}
The paper is organized as follows: in Sec. \ref{sec:background}, we present the background about control barrier functions and point out the potential infeasibility in safety-critical control.
In Sec. \ref{sec:point-wise-feasibility-input-space}, we analyze quantitatively the point-wise feasibility problem in the input-space.
In Sec. \ref{sec:point-wise-feasibility-state-space}, we show qualitatively the point-wise feasibility problem from the perspective of the state-space.
In Sec. \ref{sec:persistently-feasible-formulation}, we propose an optimal-decay form of CBF-QP and CLF-CBF-QP and prove their point-wise feasibility for any state lying strictly inside the safe set.
An adaptive cruise control example is used to numerically validate our proposed optimal-decay form in Sec. \ref{sec:example}.
Finally, Sec. \ref{sec:conclusion} provides concluding remarks. %concludes this paper.

\section{Background}
\label{sec:background}

We consider a nonlinear affine system of the form
\begin{equation}
    \dot{\mathbf{x}}(t) = f(\mathbf{x}(t)) + g(\mathbf{x}(t)) \mathbf{u}, \label{eq:affine-system}
\end{equation}
where $\mathbf{x} \in \mathbb{R}^n$, $\mathbf{u} \in \mathbb{R}^m$, with $f: \mathbb{R}^n \rightarrow \mathbb{R}^n$ and $g: \mathbb{R}^n \rightarrow \mathbb{R}^{n \times m}$ being locally Lipschitz. The system is subject to input constraints
\begin{equation}
    \mathbf{u}(t) \in \mathcal{U}_{adm}(\mathbf{x}(t)), \ \forall t \geq 0, \label{eq:input-constraint}
\end{equation}
where $\mathcal{U}_{adm}(\mathbf{x}(t)) \subset \mathbb{R}^m$ denotes the set of admissible inputs, which could be state dependent.

We consider a set $\mathcal{C} \subset \mathbb{R}^n$ defined as the zero-superlevel set of a continuously differentiable function $h: \mathbb{R}^n \rightarrow \mathbb{R}$, yielding:
\begin{equation}
\label{eq:safe-set}
    \begin{split}
        \mathcal{C} &= \{\mathbf{x} \in \mathbb{R}^n : h(\mathbf{x}) \geq 0 \}, \\
        \partial \mathcal{C} &= \{\mathbf{x} \in \mathbb{R}^n : h(\mathbf{x}) = 0 \}, \\
        \text{Int}(\mathcal{C}) &= \{\mathbf{x} \in \mathbb{R}^n : h(\mathbf{x}) > 0 \}.
    \end{split}
\end{equation}
Throughout this paper, we refer to $\mathcal{C}$ as a safe set.

The definitions of control barrier functions and control Lyapunov functions are summarized as follow, see \cite{ames2019control} for detailed explanations.
The function $h$ becomes a control barrier function if $\dfrac{\partial h}{\partial \mathbf{x}} \neq 0$ for all $\mathbf{x} \in \partial \mathcal{C}$, and there exists an extended class $\mathcal{K}_{\infty}$ function $\alpha$ such that for the control system \eqref{eq:affine-system}, $h$ satisfies
\begin{equation}
    \exists ~\mathbf{u} ~\text{s.t.} ~\dot{h}(\mathbf{x}, \mathbf{u}) \geq -\alpha(h(\mathbf{x})), ~\alpha \in \mathcal{K}_{\infty}. \label{eq:cbf-original-definition}
\end{equation}

Besides the system safety, we are also interested in stabilizing the system with a feedback control law $\mathbf{u}$ under a control Lyapunov function $V$ with a class $\mathcal{K}$ function $\gamma$, \textit{i.e.},
\begin{equation}
    \exists ~\mathbf{u} ~\text{s.t.} ~\dot{V}(\mathbf{x}, \mathbf{u}) \leq -\gamma(V(\mathbf{x})), ~\gamma \in \mathcal{K}. \label{eq:clf-original-definition}
\end{equation}

Note that we can write down
\begin{equation}
    \dot{h}(\mathbf{x}, \mathbf{u}) = L_f h(\mathbf{x}) + L_g h(\mathbf{x}) \mathbf{u},
\end{equation}
where $L_f h(\mathbf{x})$ and $L_g h(\mathbf{x})$ are Lie-derivatives of $h(\mathbf{x})$ along $f(\mathbf{x})$ and $g(\mathbf{x})$, respectively. We can also write down
\begin{equation}
    \dot{V}(\mathbf{x}, \mathbf{u}) = L_f V(\mathbf{x}) + L_g V(\mathbf{x}) \mathbf{u},
\end{equation}
where $L_f V(\mathbf{x})$ and $L_g V(\mathbf{x})$ are Lie-derivatives of $V(\mathbf{x})$.
The above construction of the CLF and CBF allows us to define safety-critical control for a nonlinear affine system \eqref{eq:affine-system}.

Given a feedback controller $\mathbf{u} = k(\mathbf{x})$ for the control system \eqref{eq:affine-system}, we wish to guarantee safety.
We consider the following Quadratic Program (QP) based controller that finds the optimal $\mathbf{u}$ in the optimization as follows:
\noindent\rule{\columnwidth}{0.5pt}
\textbf{CBF-QP:}
\begin{subequations}
\label{eq:cbf-qp}
\begin{align}
    \mathbf{u}(\mathbf{x}) & = \argmin_{\mathbf{u} \in \mathbf{R}^m}  \dfrac{1}{2} ||\mathbf{u} - k(\mathbf{x})||^2 \label{subeq:cbf-qp-cost}\\
    \text{s.t.} \ & L_f h(\mathbf{x}) + L_g h(\mathbf{x}) \mathbf{u} \geq -\alpha(h(\mathbf{x})), \label{subeq:cbf-qp-cbf-constraint} \\
    \ & \mathbf{u} \in \mathcal{U}_{adm}(\mathbf{x}). \label{subeq:cbf-qp-input-constraint}
\end{align}
\end{subequations}
\noindent\rule{\columnwidth}{0.5pt}

When the input constraint \eqref{subeq:cbf-qp-input-constraint} is excluded, we have a single inequality constraint, thus the CBF-QP has a closed-form solution per the KKT conditions, and this method was used in \cite{freeman2008robust, ames2014rapidly}. However, when the input constraint is considered, there might not exist any $\mathbf{u}$ satisfying both input constraint and CBF constraint simultaneously. This could lead to a potential infeasible optimization problem.

We could also use a QP based formulation of safety-critical control which unifies safety and stability.
Concretely, we consider the following QP based controller:
\noindent\rule{\columnwidth}{0.5pt}
\textbf{CLF-CBF-QP:}
\begin{subequations}
\label{eq:clf-cbf-qp}
\begin{align}
    \mathbf{u}(\mathbf{x}) & = \argmin_{(\mathbf{u}, \delta) \in \mathbf{R}^{m+1}} \dfrac{1}{2} \mathbf{u}^T H(\mathbf{x}) \mathbf{u} + p \delta^2 \label{subeq:clf-cbf-qp-cost} \\
    \text{s.t.} \ & L_f V(\mathbf{x}) + L_g V(\mathbf{x}) \mathbf{u} \leq -\gamma(V(\mathbf{x})) + \delta, \label{subeq:clf-cbf-qp-clf-constraint}\\ 
    \ & L_f h(\mathbf{x}) + L_g h(\mathbf{x}) \mathbf{u} \geq -\alpha(h(\mathbf{x})), \label{subeq:clf-cbf-qp-cbf-constraint} \\
    \ & \mathbf{u} \in \mathcal{U}_{adm}(\mathbf{x}). \label{subeq:clf-cbf-qp-input-constraint}
\end{align}
\end{subequations}
\noindent\rule{\columnwidth}{0.5pt}
where $H(\mathbf{x})$ is any positive definite matrix (point-wise in $\mathbf{x}$), and we have a relaxation variable $\delta$ on the CLF constraint \eqref{subeq:clf-cbf-qp-clf-constraint} with additional quadratic cost in \eqref{subeq:clf-cbf-qp-cost}.
When we exclude the input constraint \eqref{subeq:clf-cbf-qp-input-constraint} out of the optimization, the solvability can be guaranteed since the CLF constraint is relaxed and the CBF constraint \eqref{subeq:clf-cbf-qp-cbf-constraint} is the only hard constraint. This method was applied in \cite{ames2013towards, galloway2015torque}. However, when the input constraint \eqref{subeq:clf-cbf-qp-input-constraint} is also considered, we might again encounter an infeasible optimization problem.
\section{Point-wise Feasibility in Input-Space}
\label{sec:point-wise-feasibility-input-space}
Having presented the background of safety-critical control, we will now show how to pick an appropriate $\alpha$ in the CLF-CBF-QP/CBF-QP to guarantee point-wise feasibility from the perspective of input-space, \textit{i.e.}, how to guarantee feasibility of the optimization problem at a given state $\mathbf{x}(t) = \mathbf{x}_t$.

For the state $\mathbf{x}_t$ at time $t$, we define the feasible superlevel set $\mathcal{U}_{cbf}(\mathbf{x}_t)$ as the region satisfying CBF constraint in input-space, \textit{i.e.},
\begin{equation}
\label{eq:U-cbf}
    \begin{split}
        \mathcal{U}_{cbf}(\mathbf{x}_t) :=& \{\mathbf{u} \in \mathbb{R}^m: L_f h(\mathbf{x}_t) + L_g h(\mathbf{x}_t)\mathbf{u} \geq \\ &- \alpha(h(\mathbf{x}_t))\}.
    \end{split}
\end{equation}
Note that $ \mathcal{U}_{cbf}(\mathbf{x}_t)$ is a half-space in the input-space $\mathbb{R}^m$ since the CBF constraint is affine in $\mathbf{u}$. The level set of the CBF constraint in input-space is defined as $\partial\mathcal{U}_{cbf}(\mathbf{x}_t)$,
\begin{equation}
\label{eq:Int-U-cbf}
    \begin{split}
        \partial\mathcal{U}_{cbf}(\mathbf{x}_t) =&\{\mathbf{u} \in \mathbb{R}^m: L_f h(\mathbf{x}_t)+L_g h(\mathbf{x}_t)\mathbf{u}= \\ &-\alpha(h(\mathbf{x}_t))\}.
    \end{split}
\end{equation}
Then the feasibility problem becomes whether the intersection between $\mathcal{U}_{cbf}(\mathbf{x}_t)$ and $\mathcal{U}_{adm}(\mathbf{x}_t)$ as defined in \eqref{eq:input-constraint} is empty or not. If the intersection is not empty, the optimization in CBF-QP or CLF-CBF-QP is feasible at $\mathbf{x}_t$.

In order to provide a quantitative way of explaining point-wise feasibility from the perspective of input-space, We suppose that the set of admissible inputs $\mathcal{U}_{adm}(\mathbf{x}_t)$ could be described as a convex polytope defined with $r$ vertices, presented in Fig. \ref{fig:feasibility-input-space}, where each vertex is noted as $v_i(\mathbf{x}_t) \in \mathbb{R}^m$ and $i \in \{1, 2, ..., r(\mathbf{x}_t)\}$. Then we have $\mathcal{U}_{adm}(\mathbf{x}_t)$ could be written as
\begin{equation}
\label{eq:U_adm_polytope}
    \begin{split}
        \mathcal{U}_{adm}(\mathbf{x}_t) =& \{\mathbf{u} \in \mathbb{R}^m: \mathbf{u} = \sum_{i=1}^{r(\mathbf{x}_t)} \lambda_i(\mathbf{x}_t) v_i(\mathbf{x}_t), \\
        & \sum_{i=1}^{r(\mathbf{x}_t)} \lambda_i(\mathbf{x}_t) = 1,  \lambda_i(\mathbf{x}_t) \geq 0\}.
    \end{split}
\end{equation}
\begin{figure}%[!htp]
    \centering
    \includegraphics[width=1.0\linewidth]{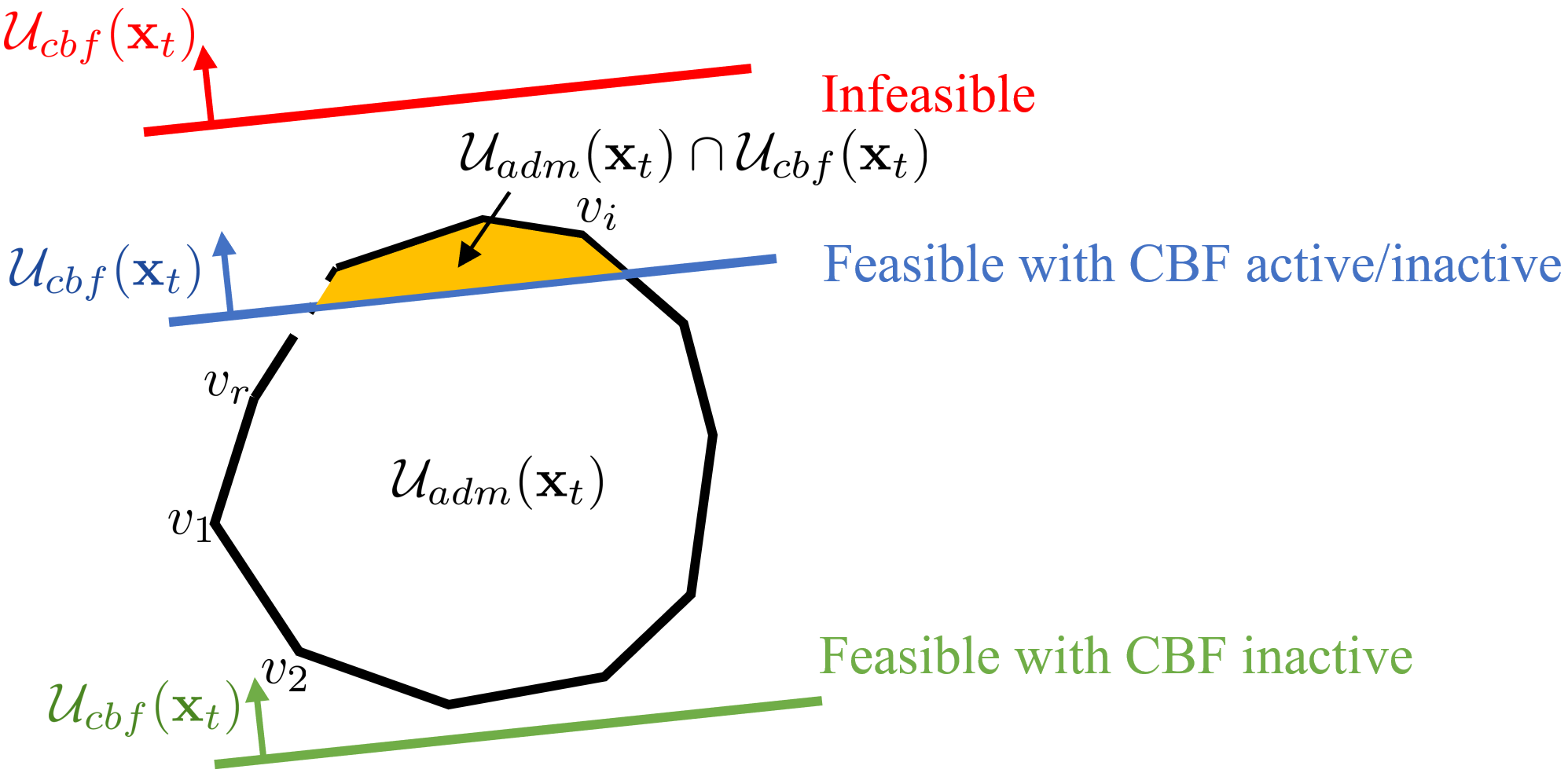}
    \caption{The set of admissible inputs $\mathcal{U}_{adm}(\mathbf{x}_t)$ is defined as a convex polytope with $r$ vertices. The upper half-spaces, indicated by the red, blue, green arrows, repsent the feasible superlevel set $\mathbf{u}_{cbf}(\mathbf{x}_t)$ corresponding to three different superlevel sets. We can see clearly that at least one vertex of polytope $\mathcal{U}_{adm}(\mathbf{x}_t)$ needs to lie inside or on the surface of the set $\mathcal{U}_{cbf}(\mathbf{x}_t)$ to guarantee the intersection set between $\mathcal{U}_{adm}(\mathbf{x}_t)$ and $\mathcal{U}_{cbf}(\mathbf{x}_t)$ is not empty.}
    \label{fig:feasibility-input-space}
\end{figure}
Then the necesary and sufficient condition of having the intersection between $\mathcal{U}_{cbf}(\mathbf{x}_t)$ defined in \eqref{eq:U-cbf} and $\mathcal{U}_{adm}(\mathbf{x}_t)$ defined in \eqref{eq:U_adm_polytope} being not empty is that, at least one vertex of polytope $\mathcal{U}_{adm}(\mathbf{x}_t)$ lies inside or on the surface of the set $\mathcal{U}_{cbf}(x)$, shown as follows,
\begin{equation}
    \mathcal{U}_{adm}(\mathbf{x}_t) \cap \mathcal{U}_{cbf}(\mathbf{x}_t) \neq \emptyset \Leftrightarrow \exists i, v_i \in \mathcal{U}_{cbf}(\mathbf{x}_t).
\end{equation}
Hence, the set of candidate $\alpha$ functions in the CBF constraint \eqref{eq:cbf-original-definition}, denoted as $\mathcal{K}_{fea}^{\alpha}(\mathbf{x}_t)$, that guarantees the feasibility of the optimization, becomes the union of the set of functions that guarantees any vertex in $\mathcal{U}_{adm}(\mathbf{x}_t)$ lies inside or on the surface of the set $\mathcal{U}_{cbf}(\mathbf{x}_t)$, which could be written as follows
\begin{equation*}
    \begin{split}
    \mathcal{K}_{fea}^{\alpha}(\mathbf{x}_t) =& \bigcup_{i=1}^{r(\mathbf{x}_t)} \{\alpha \in \mathcal{K}_{\infty}: \alpha(h(\mathbf{x}_t)) \geq\\
    & - L_f h(\mathbf{x}_t) - L_g h(\mathbf{x}_t) v_i(\mathbf{x}_t) \},
    \end{split}
\end{equation*}
and it could be reformulated with a minimum operator
\begin{equation}
    \label{eq:K-feas-alpha}
    \begin{split}
        \mathcal{K}_{fea}^{\alpha}(\mathbf{x}_t) =& \{\alpha \in \mathcal{K}_{\infty} : \alpha(h(\mathbf{x}_t)) \geq \\
        & \min_{i} (- L_f h(\mathbf{x}_t) - L_g h(\mathbf{x}_t) v_i(\mathbf{x}_t)) \}.
    \end{split}
\end{equation}
Therefore, given a function $\alpha$, the CBF-QP/CLF-CBF-QP in \eqref{eq:cbf-qp} and \eqref{eq:clf-cbf-qp} are point-wise feasible if $\exists \alpha \in \mathcal{K}_{\infty}$ s.t.,
\begin{equation*}
    \alpha(h(\mathbf{x}_t)) \geq \min_{i} (- L_f h(\mathbf{x}_t) - L_g h(\mathbf{x}_t) v_i(\mathbf{x}_t))
\end{equation*}
is satisfied at state $\mathbf{x}_t$.

Moreover, given a state $\mathbf{x}_t$, from \eqref{eq:K-feas-alpha}, we can see when 
\begin{equation}
\label{eq:input-space-always-point-wise-feasible}
    \min_{i} (- L_f h(\mathbf{x}_t) - L_g h(\mathbf{x}_t) v_i(\mathbf{x}_t)) \leq 0,
\end{equation}
we have $\mathcal{K}_{fea}^{\alpha}(\mathbf{x}_t) = \mathcal{K}_{\infty}$, \textit{i.e.}, we have the point-wise feasibility, for any $\alpha \in \mathcal{K}_{\infty}$. However, when \eqref{eq:input-space-always-point-wise-feasible} is not satisfied,
$\mathcal{K}_{fea}^{\alpha}(\mathbf{x}_t)$ becomes a proper subset of $\mathcal{K}_{\infty}$ and the function $\alpha$ has a lower bound at state $\mathbf{x}_t$ to guarantee the point-wise feasibility.

From above, we could pick an appropriate $\alpha$ function to guarantee point-wise feasibility at a given state $\mathbf{x}_t$. Beside feasibility, we are also interested in whether the CBF constraint is activated during the optimization.

\begin{remark}
We have $\forall i, v_i(\mathbf{x}_t) \in \mathcal{U}_{cbf}(\mathbf{x}_t)$ when
\begin{equation}
\label{eq:input-space-cbf-confine}
    \alpha(h(\mathbf{x}_t)) \geq \max_{i} (- L_f h(\mathbf{x}_t) - L_g h(\mathbf{x}_t) v_i(\mathbf{x}_t)),
\end{equation}
which means $\mathcal{U}_{adm}(\mathbf{x}_t) \subset \mathcal{U}_{cbf}(\mathbf{x}_t)$, therefore in this case, the CBF constraint does not confine the input constraint during the optimization. This case is illustrated by the green $\mathcal{U}_{cbf}(\mathbf{x}_t)$ in Fig. \ref{fig:feasibility-input-space}. When \eqref{eq:input-space-cbf-confine} is not satisfied, we have the CBF constraint confine the input constraint, \textit{i.e.}, the intersection between $\mathcal{U}_{adm}(\mathbf{x}_t)$ and $\mathcal{U}_{cbf}(\mathbf{x}_t)$ becomes a proper set of $\mathcal{U}_{adm}(\mathbf{x}_t)$, indicated by the blue $\mathcal{U}_{cbf}(\mathbf{x}_t)$ in Fig. \ref{fig:feasibility-input-space}.
When the intersection is empty, illustrated by the red $\mathcal{U}_{cbf}(\mathbf{x}_t)$ in Fig. \ref{fig:feasibility-input-space}, the optimization problem is infeasible.
\end{remark}

\begin{remark}
\label{remark:input-space-cbf-activation}
We say a constraint is active in the optimization when the optimal solution lies on the constraint line \cite{boyd2004convex}. When the CBF constraint does not confine the input constraint, the CBF constraint becomes inactive during the optimization. However, when the CBF constraint confines the input constraint, it does not necessarily guarantee the CBF constraint activation, as the value of optimal solution also depends on the design of cost function in the optimization. Therefore, the choice of $\alpha$ function together with the design of cost function determine the activation of CBF constraint in the safety-critical optimal control, shown in Fig. \ref{fig:feasibility-input-space}.
\end{remark}

\begin{remark}
\label{remakr:set-property-bounded-unbounded}
In this section, the admissible set is assumed as a bounded convex polytope. In fact, these discussions could be easily generalized for an unbounded convex set, as any unbounded set could be regarded as a limit of a sequence of bounded sets \cite[Chap. 11]{bartle2000introduction}. Moreover, our proposed optimal-decay approach in this paper that guarantees point-wise feasibility will only rely on the admissible set being convex while not strictly assuming it as a bounded convex polytope, see Theorem \ref{thm:point-wise-feasibility}.
\end{remark}

We have seen point-wise feasibility in input-space in this section.
Next, we will look at point-wise feasibility in the state-space in Sec. \ref{sec:point-wise-feasibility-state-space}.
After that, we provide a formulation which allows us to guarantee the point-wise feasibility without parameter tuning in Sec. \ref{sec:persistently-feasible-formulation}.

\section{Point-wise Feasibility in State-Space}
\label{sec:point-wise-feasibility-state-space}
\begin{figure*}
    \centering
    \begin{subfigure}[t]{0.24\linewidth}
        \centering
        \includegraphics[height=3.8cm]{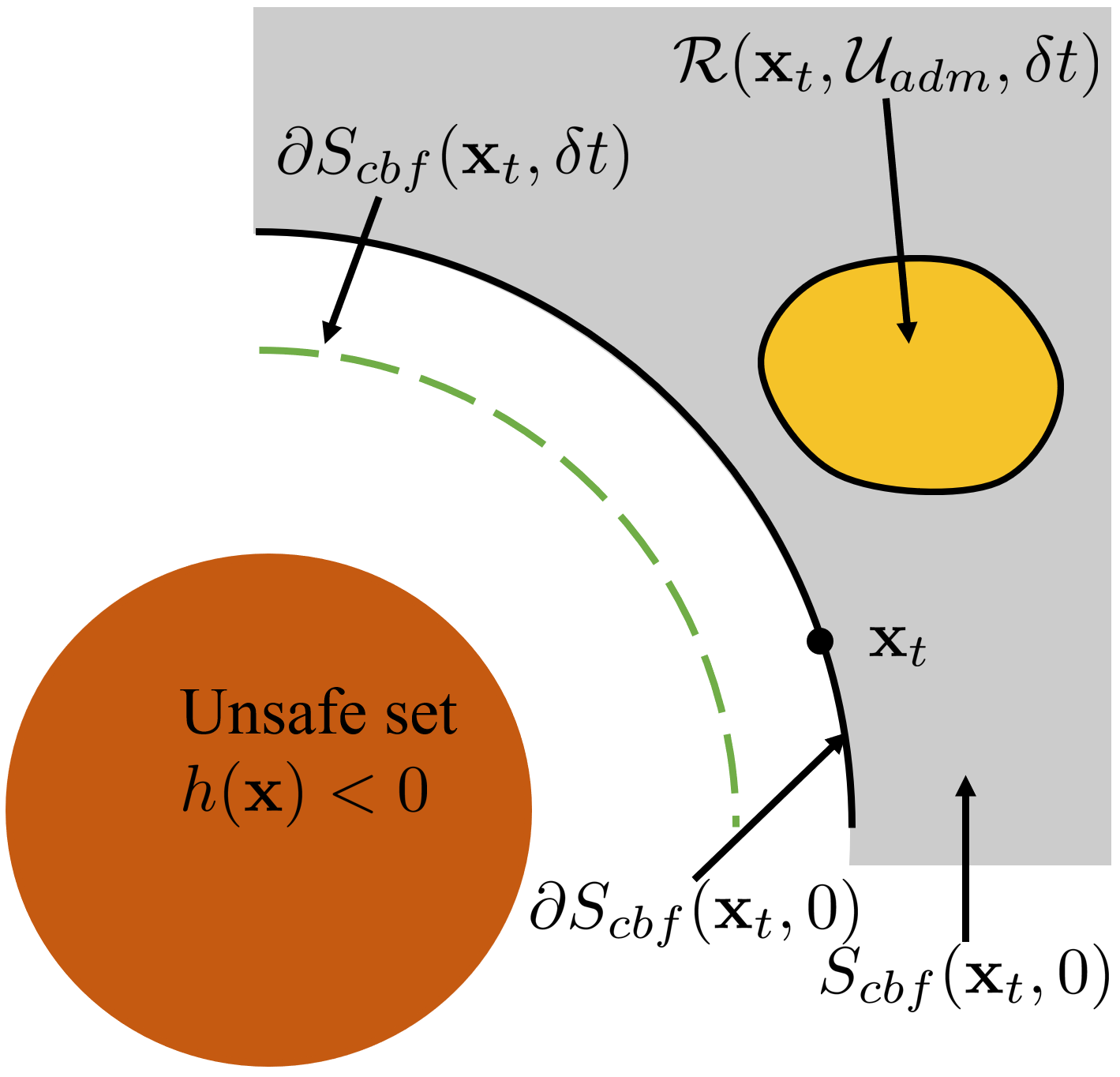}
        \caption{}
        \label{subfig:feasibility-far-away}
    \end{subfigure}
    \begin{subfigure}[t]{0.24\linewidth}
        \centering
        \includegraphics[height=3.8cm]{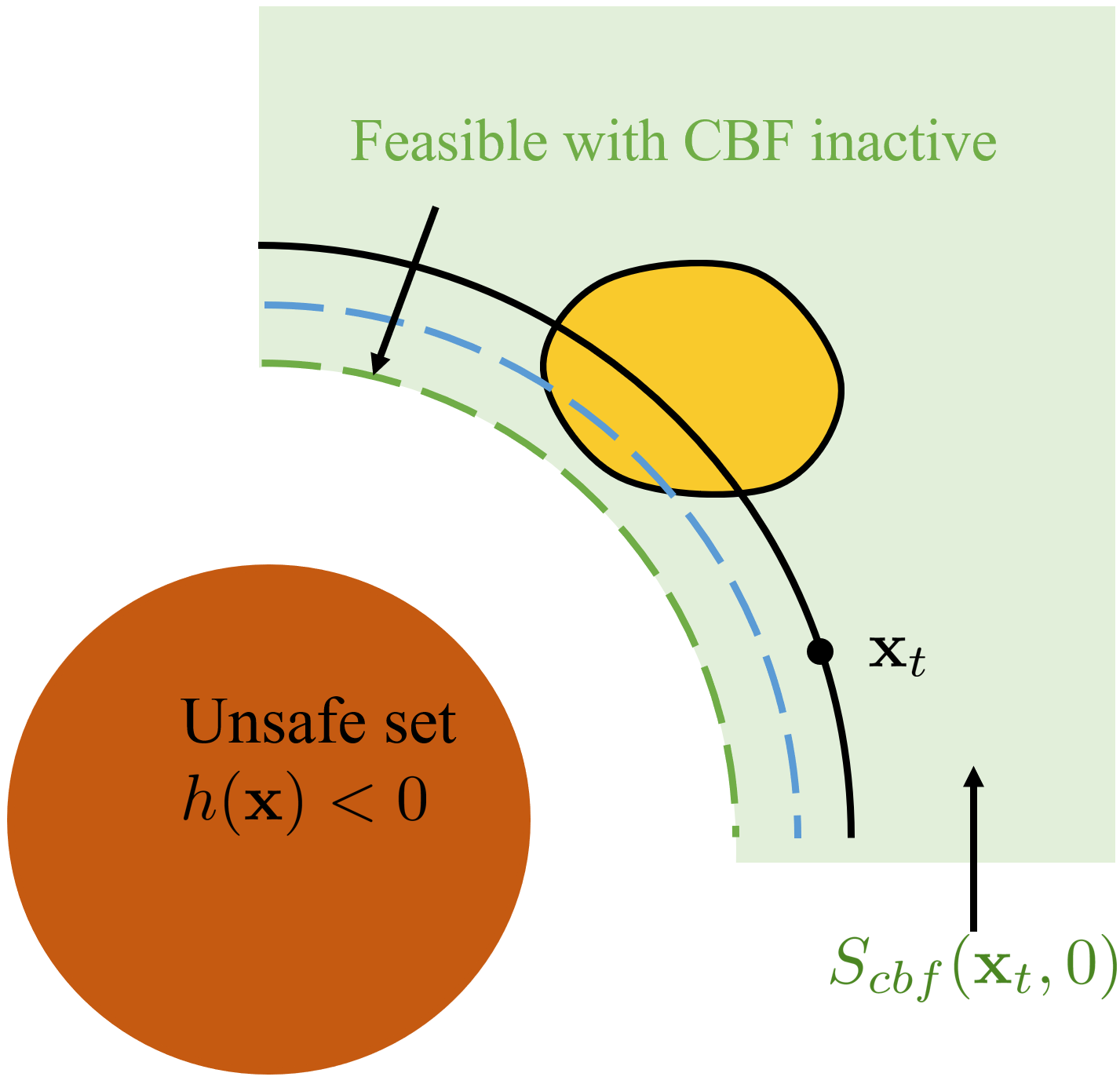}
        \caption{}
        \label{subfig:feasibilty-around}
    \end{subfigure}
    \begin{subfigure}[t]{0.24\linewidth}
        \centering
        \includegraphics[height=3.8cm]{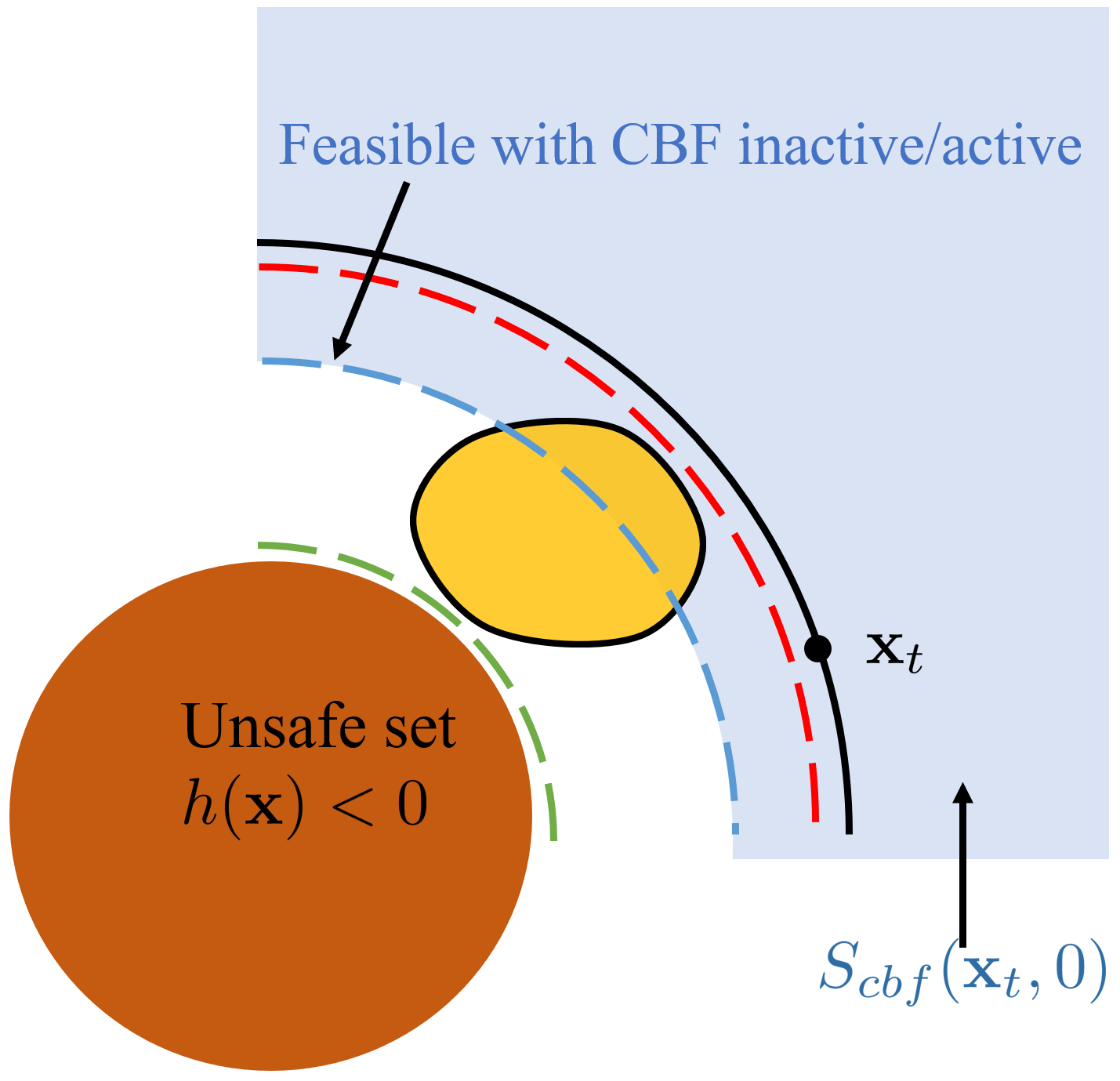}
        \caption{}
        \label{subfig:feasibility-closer}
    \end{subfigure}
    \begin{subfigure}[t]{0.24\linewidth}
        \centering
        \includegraphics[height=3.8cm]{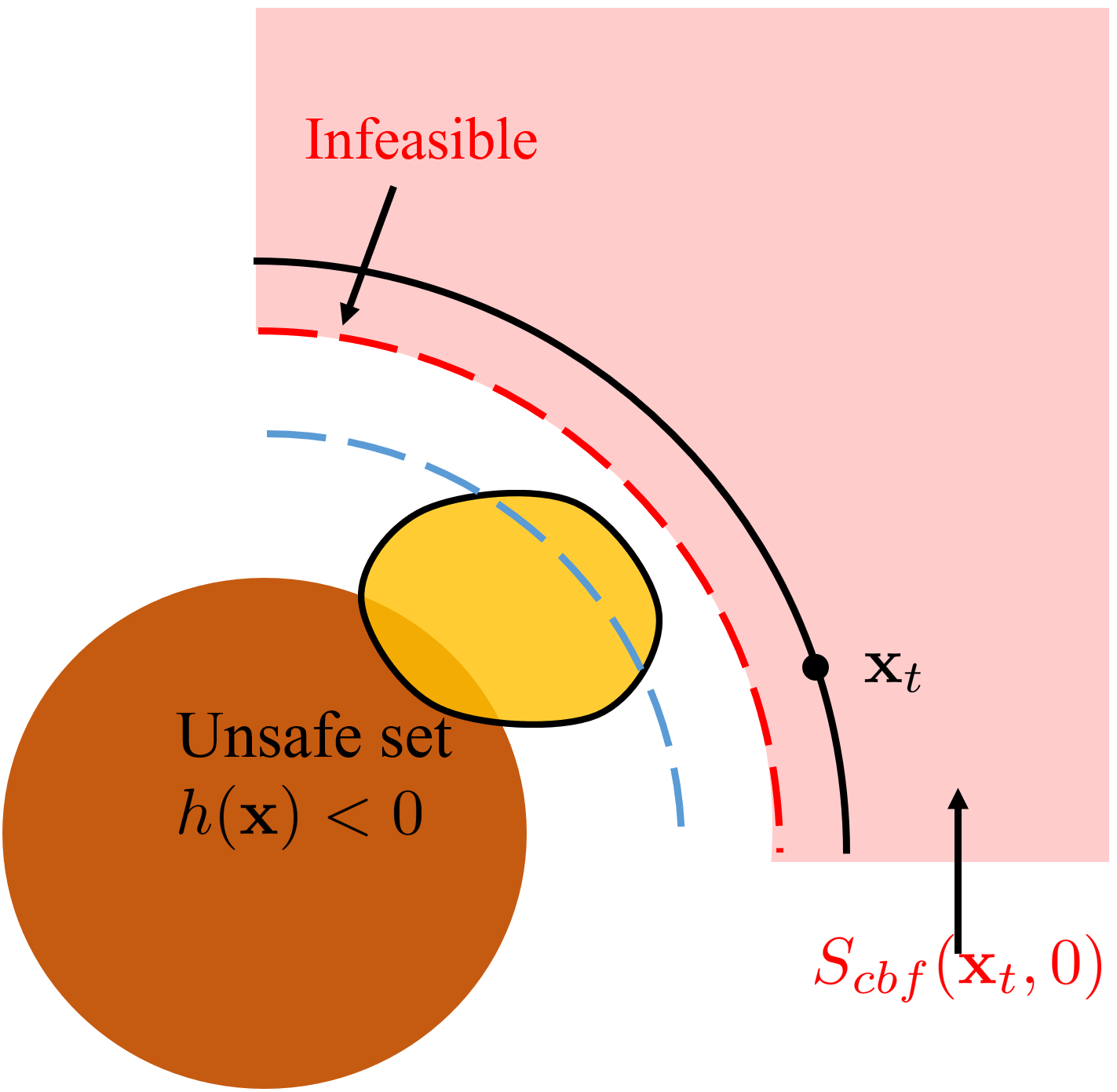}
        \caption{}
        \label{subfig:feasibility-very-close}
    \end{subfigure}
    \caption{Point-wise feasibility problem in different scenarios. Given a state $\mathbf{x}_t$, we illustrate $\mathcal{R}(\mathbf{x}_t, \mathcal{U}_{adm}, \delta t)$ as a closed region of yellow color and $\partial\mathcal{S}_{cbf}(\mathbf{x}_t, 0)$ as level sets in black. $\partial\mathcal{S}_{cbf}(\mathbf{x}_t, \delta t)$ are plotted in red (infeasible), blue (feasible with active/inactive CBF constraint) and green (feasible with inactive CBF constraint) for different scenarios. $\mathcal{S}_{cbf}(\mathbf{x}_t, 0)$ (grey) and $\mathcal{S}_{cbf}(\mathbf{x}_t, \delta t)$ (green, blue, red) are regions on the top-right side of their corresponding level sets.}
    \label{fig:feasibility-state-space}
\end{figure*}
The point-wise feasibility problem could also be understood through qualitative illustration in the state-space.
Given a state $\mathbf{x}_t$ at time $t$, we will define $\mathcal{R}(\mathbf{x}_t, \mathcal{U}_{adm}, \delta t)$ to represent the set of reachable states after infinitesimal time $\delta t$ while satisfying system dynamics \eqref{eq:affine-system} and input constraint $\mathcal{U}_{adm}(\mathbf{x}_t)$ starting from state $\mathbf{x}(t) = \mathbf{x}_t$, \textit{i.e.},
\begin{equation}
\label{eq:reachable-set-infinitesimal-time}
    \begin{split}
        \mathcal{R}(\mathbf{x}_t, \mathcal{U}_{adm}, \delta t) = \{\mathbf{x}(t+\delta t) \in \mathbb{R}^n: \forall \bar{t} \in [t, t + \delta], \\
        \dot{\mathbf{x}}(\bar{t}) = f(\mathbf{x}(\bar{t})) + g(\mathbf{x}(\bar{t})) \mathbf{u}(\bar{t}), \\ 
        \mathbf{u}(\bar{t}) \in \mathcal{U}_{adm}(\mathbf{x}(\bar{t})), \mathbf{x}(t) = \mathbf{x}_t\}.
    \end{split}
\end{equation}
The evolution of the system dynamics also needs to be safe and thus satisfy the definition of control barrier function \eqref{eq:cbf-original-definition} during time segment $[t, t + \delta t]$, thus we have
\begin{equation}
    h(\mathbf{x}(t+\delta t)) \geq h(\mathbf{x}(t)) - \int_{t}^{t+\delta t} \alpha(h(\mathbf{x}(\bar{t}))) d\bar{t}.
\end{equation}
This allows us to define the superlevel set in state-space for $\mathbf{x}(t+\delta t)$ satisfying the CBF constraints,
\begin{equation}
    \begin{split}
        \mathcal{S}_{cbf}&(\mathbf{x}_t, \delta t) = \{\mathbf{x} \in \mathbb{R}^n: h(\mathbf{x})
         \geq h(\mathbf{x}_t) \\ & - \int_{t}^{t+\delta t} \alpha(h(\mathbf{x}(\bar{t}))) d\bar{t}, \mathbf{x}(t) = \mathbf{x}_t\}.
    \end{split}
\end{equation}
We also define
\begin{equation}
    \mathcal{S}_{cbf}(\mathbf{x}_t, 0) = \{\mathbf{x} \in \mathbb{R}^n:  h(\mathbf{x}) \geq h(\mathbf{x}_t)\},
\end{equation}
motivated by $\int_{t}^{t+\delta t} \alpha(h(\mathbf{x}(\bar{t}))) d\bar{t} = 0$ when $\delta t = 0$.
The set $\mathcal{S}_{cbf}(\mathbf{x}_t, 0)$ corresponds to the set of all possible $\mathbf{x}$ for which $h(\mathbf{x}) \geq h(\mathbf{x}_t)$.

Since the state $\mathbf{x}(t + \delta t)$ should satisfy the system dynamics and control barrier function constraint, the optimization problem is then point-wise feasible at state $\mathbf{x}_t$ when the intersection between $\mathcal{R}(\mathbf{x}_t, \mathcal{U}_{adm}, \delta t)$ and $\mathcal{S}_{cbf}(\mathbf{x}_t, \delta t)$ is not empty.
We are interested in whether the intersection between $\mathcal{R}(\mathbf{x}_t, \mathcal{U}_{adm}, \delta t)$ and $\mathcal{S}_{cbf}(\mathbf{x}_t, \delta t)$ is empty or not under different circumstances. However, it is numerically complicated or even impossible to calculate $\mathcal{R}(\mathbf{x}_t, \mathcal{U}_{adm}, \delta t)$ for a general nonlinear affine system. Thus, we provide an intuition of understanding point-wise feasibility problem through geometry in state-space.

\begin{remark}
Notice that we always have $\mathcal{S}_{cbf}(\mathbf{x}_t, 0) \subset \mathcal{S}_{cbf}(\mathbf{x}_t, \delta t)$ as $\alpha$ is class $\mathcal{K}_{\infty}$ function and $h(.)$ is positive.
\end{remark}

In practice, we usually define the safety set $\mathcal{C}$ in \eqref{eq:safe-set} corresponding to the free space outside the obstacle, illustrated in Fig. \ref{fig:feasibility-state-space}. $\partial\mathcal{S}_{cbf}(\mathbf{x}_t, 0)$ and $\partial\mathcal{S}_{cbf}(\mathbf{x}_t, \delta t)$ are illustrated with black 
solid and colorful dashed curves respectively, and $\mathcal{S}_{cbf}(\mathbf{x}_t, \delta t)$, $\mathcal{S}_{cbf}(\mathbf{x}_t, 0)$ are illustrated as the regions on the top-right side of them. Since we always have $\mathcal{S}_{cbf}(\mathbf{x}_t, 0) \subset \mathcal{S}_{cbf}(\mathbf{x}_t, \delta t)$, $\partial\mathcal{S}_{cbf}(\mathbf{x}_t, \delta t)$ is always closer to the obstacle, lying on the bottom-left side of $\partial\mathcal{S}_{cbf}(\mathbf{x}_t, 0)$ for any choice of $\alpha$ function. We classify the point-wise feasibility problem into three scenarios as follows.

\subsection{Moving away from obstacles}
When $\mathcal{R}(\mathbf{x}_t, \mathcal{U}_{adm}, \delta t) \subset \mathcal{S}_{cbf}(\mathbf{x}_t, 0)$, \textit{i.e.}, the system is moving away from obstacles. The scenario is illustrated in Fig. \ref{subfig:feasibility-far-away}. In this scenario, we have $\mathcal{R}(\mathbf{x}_t, \mathcal{U}_{adm}, \delta t) \subset \mathcal{S}_{cbf}(\mathbf{x}_t, 0) \subset \mathcal{S}_{cbf}(\mathbf{x}_t, \delta t)$. This means that for any class $\mathcal{K}_{\infty}$ function $\alpha$, the optimization will always be point-wise feasible at state $\mathbf{x}_t$.

\subsection{Moving around obstacles}
When $\mathcal{R}(\mathbf{x}_t, \mathcal{U}_{adm}, \delta t)$ intersects with $\partial \mathcal{S}_{cbf}(\mathbf{x}_t, 0)$, we have the reachable state-space lies partly on the top-right side of $\partial \mathcal{S}_{cbf}(\mathbf{x}_t, 0)$, shown in Fig. \ref{subfig:feasibilty-around}. In this scenario, the optimization is always feasible for any class $\mathcal{K}_{\infty}$ function $\alpha$, as $\mathcal{R}(\mathbf{x}_t, \mathcal{U}_{adm}, \delta t) \cap \mathcal{S}_{cbf}(\mathbf{x}_t, 0)$ is not empty and is always a subset of $\mathcal{S}_{cbf}(\mathbf{x}_t, \delta t)$.

\subsection{Moving close to obstacles}
When $\mathcal{R}(\mathbf{x}_t, \mathcal{U}_{adm}, \delta t) \cap \mathcal{S}_{cbf}(\mathbf{x}_t, 0) = \emptyset$, this usually happens when the system is moving close to obstacles, shown in Fig. \ref{subfig:feasibility-closer} and \ref{subfig:feasibility-very-close}. In this scenario, when $\alpha$ becomes too small (CBF level set in red), the intersection between $\mathcal{R}(\mathbf{x}_t, \mathcal{U}_{adm}, \delta t)$ and $\mathcal{S}_{cbf}(\mathbf{x}_t, \delta t)$ is empty and the optimization problem becomes infeasible at state $\mathbf{x}_t$. This indicates that $\alpha(h(\mathbf{x}_t))$ needs to be greater than a lower bound to make the optimization point-wise feasible at state $\mathbf{x}_t$.

\begin{remark}
\label{remark:space-space-cbf-activation}
When $\mathcal{R}(\mathbf{x}_t, \mathcal{U}_{adm}, \delta t) \subset \mathcal{S}_{cbf}(\mathbf{x}_t, \delta t)$, the CBF constraint does not confine the reachable set which means the CBF constraint is inactive in the optimization, shown with green level sets $\mathcal{S}_{cbf}(\mathbf{x}_t, \delta t)$ in Fig. \ref{fig:feasibility-state-space}.
When the intersection between $\mathcal{R}(\mathbf{x}_t, \mathcal{U}_{adm}, \delta t)$ and $\mathcal{S}_{cbf}(\mathbf{x}_t, \delta t)$ is non-empty and becomes as a proper subset of $\mathcal{R}(\mathbf{x}_t, \mathcal{U}_{adm}, \delta t)$ shown with blue level sets $\mathcal{S}_{cbf}(\mathbf{x}_t, 0)$ in Fig. \ref{fig:feasibility-state-space}, the CBF constraint does confine the reachable set.
This does not necessarily guarantee CBF constraint activation, as the constraint activation also depends on the design of cost function, which is similar to what we discussed in Remark \ref{remark:input-space-cbf-activation}.
For the red level sets $\mathcal{S}_{cbf}(\mathbf{x}_t, 0)$ in Fig. \ref{subfig:feasibility-closer} and \ref{subfig:feasibility-very-close}, the optimization is infeasible.
\end{remark}

\section{Point-wise Feasible Formulation}
\label{sec:persistently-feasible-formulation}
\subsection{Formulation and Point-wise Feasibility}
In Sec. \ref{sec:point-wise-feasibility-input-space} and \ref{sec:point-wise-feasibility-state-space}, we have seen that the optimization in safety-critical control might become point-wise infeasible at a given state $\mathbf{x}_t$, if either $\mathcal{U}_{cbf}(\mathbf{x}_t) \cap \mathcal{U}_{adm}(\mathbf{x}_t) = \emptyset$ or $\mathcal{R}(\mathbf{x}_t, \mathcal{U}_{adm}, \delta t) \cap \mathcal{S}_{cbf}(\mathbf{x}_t, \delta t) = \emptyset$ which means the convergence of control barrier function is less than a lower bound. In other words, when the decay rate of the lower bound of $h(\mathbf{x})$ is not large enough, we might encounter infeasibility in the optimization problem. To solve this problem, we could manually tune the form of $\alpha$ function to make the optimization feasible, however, this tuning process becomes relatively difficult when the system dynamics becomes complicated. This motivates us to introduce an optimal-decay form of CBF constraint to guarantee the point-wise feasibility for any $\mathbf{x}$ with $h(\mathbf{x}) > 0$.

With the same notation as CBF-QP in \eqref{eq:cbf-qp} and CLF-CBF-QP in \eqref{eq:clf-cbf-qp}, we introduce an optimal-decay form of CBF-QP and CLF-CBF-QP in this section. The optimal decay CBF-QP is formulated as follows

\noindent\rule{\columnwidth}{0.5pt}
\textbf{Optimal-decay CBF-QP:}
\begin{subequations}
\label{eq:soft-cbf-qp}
\begin{align}
    \mathbf{u}(\mathbf{x}) & = \argmin_{(\mathbf{u}, \omega) \in \mathbf{R}^{m+1}}  \dfrac{1}{2} ||\mathbf{u} - k(\mathbf{x})||^2 + p_{\omega} (\omega - \omega_0)^2 \label{subeq:soft-cbf-qp-cost}\\
    \text{s.t.} \ & L_f h(\mathbf{x}) + L_g h(\mathbf{x}) \mathbf{u} \geq - \omega \alpha(h(\mathbf{x})), \label{subeq:soft-cbf-qp-cbf-constraint} \\
    \ & \mathbf{u} \in \mathcal{U}_{adm}(\mathbf{x}). \label{subeq:soft-cbf-qp-input-constraint}
\end{align}
\end{subequations}
\noindent\rule{\columnwidth}{0.5pt}
Compared with CBF-QP, we optimize the decay rate of the CBF constraint with a new variable $\omega$ in \eqref{subeq:soft-cbf-qp-cbf-constraint} and add a quadratic cost in \eqref{subeq:soft-cbf-qp-cost}. $p_{\omega}$ is a positive scalar and a scalar $\omega_0$ could usually be chosen to tune the performance of the controller. Similarly, we develop the optimal form of CLF-CBF-QP as follows

\noindent\rule{\columnwidth}{0.5pt}
\textbf{Optimal-decay CLF-CBF-QP:}
\begin{subequations}
\label{eq:soft-clf-cbf-qp}
\begin{align}
    \mathbf{u}(\mathbf{x}) &= \argmin_{(\mathbf{u}, \delta, \omega) \in \mathbf{R}^{m+2}} \dfrac{1}{2} \mathbf{u}^T H(\mathbf{x}) \mathbf{u} + p \delta^2 + p_{\omega} (\omega - \omega_0)^2 \label{subeq:soft-clf-cbf-qp-cost} \\
    \text{s.t.} \ & L_f V(\mathbf{x}) + L_g V(\mathbf{x}) \mathbf{u} \leq -\gamma(V(\mathbf{x})) + \delta, \label{subeq:soft-clf-cbf-qp-clf-constraint}\\ 
    \ & L_f h(\mathbf{x}) + L_g h(\mathbf{x}) \mathbf{u} \geq - \omega\alpha(h(\mathbf{x})), \label{subeq:soft-clf-cbf-qp-cbf-constraint} \\
    \ & \mathbf{u} \in \mathcal{U}_{adm}(\mathbf{x}). \label{subeq:soft-clf-cbf-qp-input-constraint}
\end{align}
\end{subequations}
\noindent\rule{\columnwidth}{0.5pt}
We call \eqref{eq:soft-cbf-qp} and \eqref{eq:soft-clf-cbf-qp} as optimal-decay form of CBF-QP and CLF-CBF-QP, since it actually optimizes the decay rate of lower bound of the control barrier function with variable $\omega$ in the optimization. The point-wise feasibility of optimization for any $\mathbf{x} \in \text{Int}(\mathcal{C})$ in \eqref{eq:soft-cbf-qp}, \eqref{eq:soft-clf-cbf-qp} is illustrated with the following theorem.

\begin{theorem}
\label{thm:point-wise-feasibility}
When $\mathcal{U}_{adm}(\mathbf{x})$ is convex, the optimizations in optimal-decay CBF-QP and CLF-CBF-QP are point-wise feasible for any $\mathbf{x}$ lying inside $\mathcal{C}$, \textit{i.e.}, the optimizations are solvable with a unique solution for any $\mathbf{x}$ when $h(\mathbf{x}) > 0$.
\end{theorem}
\begin{proof}
For optimal-decay CBF-QP, since $\omega \in \mathbb{R}$ is a variable to optimize and $h(\mathbf{x}) > 0$, $-\omega \alpha(h(\mathbf{x}))$ could vary between $-\infty$ and $+\infty$ for any $\mathbf{x}$. Then, for any $\mathbf{x}$. there always exists a value $\omega$ such that the feasible region between \eqref{subeq:soft-cbf-qp-cbf-constraint} and \eqref{subeq:soft-cbf-qp-input-constraint} is non-empty. 

Moreover, we notice that the CBF constraints and input constraints are convex with respect to optimization variables $\mathbf{u}$ and $\omega$. Hence, the feasible region in the optimal-decay CBF-QP is convex and not empty. Additionally, our cost function \eqref{subeq:soft-cbf-qp-cost} is a quadratic form and positive-definite, which is strictly convex. Therefore, the minimization in \eqref{eq:soft-cbf-qp} is convex (convex cost and constraints) with non-empty feasible region, which is solvable and holds a unique solution \cite[Chap. 5]{boyd2004convex}.

Compared to optimal-decay CBF-QP, the optimal-decay CLF-CBF-QP is nothing different except there is a CLF constraint in \eqref{subeq:soft-clf-cbf-qp-clf-constraint}, which is also convex with respect to optimization variables and is optimized with the relaxation variable $\delta$. Therefore, the optimization in optimal CLF-CBF-QP is also convex with a non-empty feasible region, which is feasible with an unique solution.
\end{proof}
In fact, we could generalize the point-wise feasibility in Thm. \ref{thm:point-wise-feasibility} for any $\mathbf{x}$ when $h(\mathbf{x}) \neq 0$, as $\omega\alpha(h(\mathbf{x}))$ could vary between $-\infty$ and $\infty$ when $h(\mathbf{x}) \neq 0$ with $\omega \in \mathbb{R}$. However, this is not interesting since the system is already unsafe when $h(\mathbf{x}) < 0$.

\begin{remark}
The optimization variable $\omega$ is not required to be positive.
In fact, when the optimized value of $\omega$ is negative, this implies that the control barrier function is increasing with the optimized control input, correpsonding to the safety with set invariance.
In the scenario of $\omega$ as negative, the system is keeping further away from the obstacles.
\end{remark}
\subsection{Convergence of CBF constraint}
\label{subsec:convergence-hyperparameter}
The variable $\omega$ optimizes the convergence rate of the CBF constraint, with $\omega\alpha(.)$ still being a class $\mathcal{K}_{\infty}$ function. Therefore, with optimal-decay CBF-QP and CLF-CBF-QP, we equivalently use a state-dependent rate for the convergence of the CBF constraint, since the optimized value of $\omega$ will be calculated differently at different state $\mathbf{x}$. 

Moreover, smaller $\omega_0$ and larger $p_{\omega}$ would make the control barrier function decay slower. This makes sense from a mathematical perspective, since the larger $\omega_0$ optimizes the nominal CBF constraint with larger decay rate. Similarly, for smaller $p_{\omega}$, we will have optimized value of $\omega$ deviating more from the value $\omega_0$, which also brings a larger decay rate on the control barrier function.

\subsection{Persistent Feasibility and Safety}
\label{subsec:persistent-feasibility-safety}
Note that point-wise feasibility for all $\mathbf{x}$ with $h(\mathbf{x}) > 0$ does not guarantee that our system is persistently feasible, \textit{i.e.}, the system might become unsafe even if the system is initialized at $\mathbf{x}(0)=\mathbf{x}_0$ with $h(\mathbf{x}_0) > 0$.
This is due to the fact that when we reach $h(\mathbf{x})=0$, we can't change $\omega$ to make the optimization problem feasible while simultaneously enforcing the optimal-decay CBF constraint \eqref{subeq:soft-clf-cbf-qp-cbf-constraint} and the input constraint.
This makes the control barrier function invalid with respect to the safe set defined in \eqref{eq:safe-set} due to the input constraint.
In this case, the CBF-based control policy isn't safe, i.e., it is not forward complete, see \cite[Def. 1]{ames2019control}.
However, given any invariant set that is a subset of the safe set $\mathcal{C}$, our optimal-decay CBF-QP/CLF-CBF-QP will be persistently feasible if the initial state lies inside this safe invariant set.

The maximum invariant set for an optimal-decay controller could be different with respect to hyperparameters, such as $p_{\omega}$ and $\omega_0$, which will be illustrated in Sec. \ref{subsec:result-safety}.
Notice that solving an explicit representation of this control invariant set for a nonlinear system is usually intractable \cite{rakovic2010parameterized, athanasopoulos2014construction, rungger2017computing} and this exceeds the scope of this paper.
We will discuss the invariant safety and related performance in future work.

\begin{figure*}
    \centering
    \begin{subfigure}[t]{0.24\linewidth}
        \centering
        \includegraphics[width=0.99\linewidth]{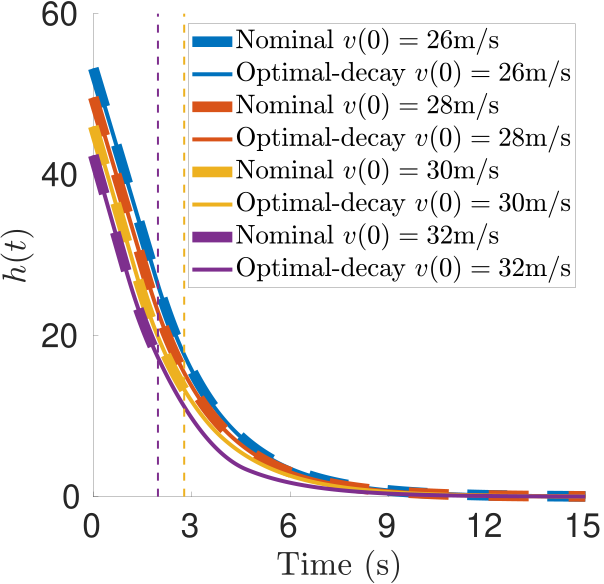}
        \caption{Control barrier function}
        \label{subfig:benchmark-feasibility-cbf}
    \end{subfigure}
    \begin{subfigure}[t]{0.24\linewidth}
        \centering
        \includegraphics[width=0.99\linewidth]{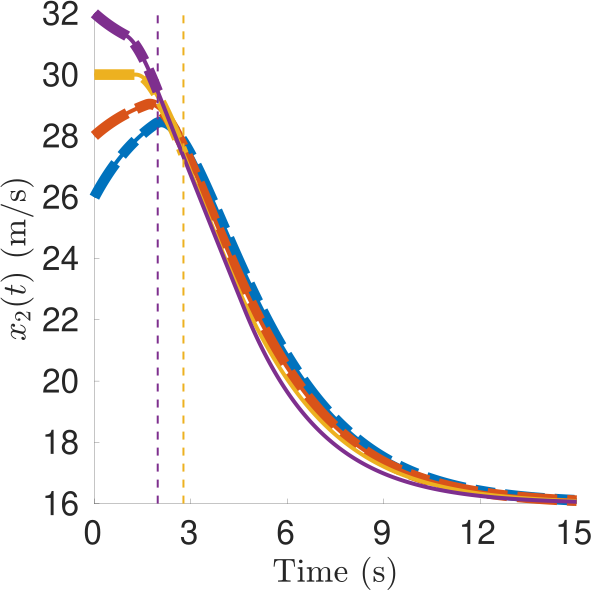}
        \caption{Ego vehicle's speed}
        \label{subfig:benchmark-feasibility-velocity}
    \end{subfigure}
    \begin{subfigure}[t]{0.24\linewidth}
        \centering
        \includegraphics[width=0.99\linewidth]{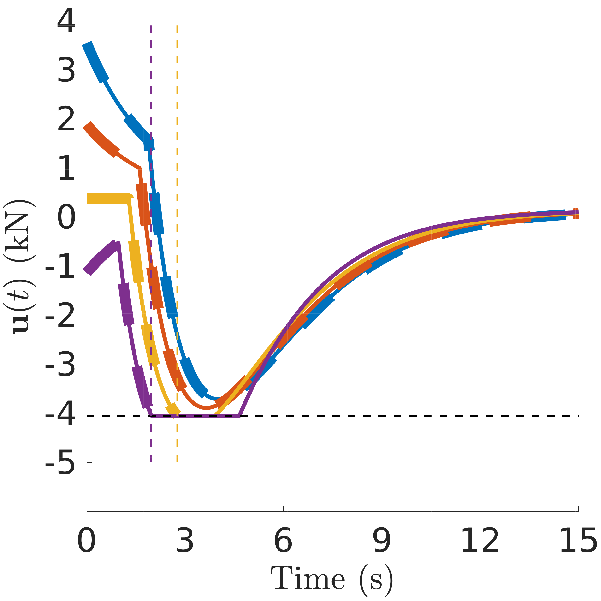}
        \caption{Optimal control input}
        \label{subfig:benchmark-feasibility-control-input}
    \end{subfigure}
    \begin{subfigure}[t]{0.24\linewidth}
        \centering
        \includegraphics[width=0.99\linewidth]{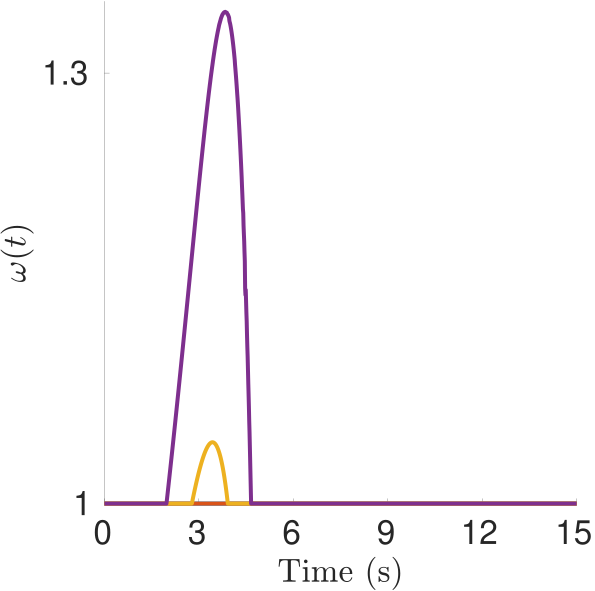}
        \caption{Optimal-decay variable}
        \label{subfig:benchmark-feasibility-omega}
    \end{subfigure}
    \caption{Simulation results of adaptive cruise control using original/optimal-decay form of CLF-CBF-QP with different initial velocities.
    The solid and dashed lines correspond to simulation with optimal-decay CLF-CBF-QP and original CLF-CBF-QP, respectively. Different colors represent different initial conditions. The colorful vertical dashed lines in Fig. \ref{subfig:benchmark-feasibility-velocity} represents when the original CLF-CBF-QP becomes infeasible and we can clearly see that the original CLF-CBF-QP becomes infeasible for purple ($x_2(0) = 32\text{m/s}$) and yellow ($x_2(0) = 30\text{m/s}$) initial conditions.
    }
    \label{fig:benchmark-different-initial-conditions}
\end{figure*}

\begin{remark}
\label{rem:p-omega-infinity}
When $p_\omega$ tends to $+\infty$ in \eqref{eq:soft-cbf-qp} and \eqref{eq:soft-clf-cbf-qp}, then given a state $\mathbf{x}$, if the original CBF constraint is satisfied, \textit{i.e.},
\begin{equation*}
    \exists \mathbf{u} \in \mathcal{U}_{adm}(\mathbf{x}), \ L_f h(\mathbf{x}) + L_g h(\mathbf{x}) \mathbf{u} \geq - \alpha(h(\mathbf{x})),
\end{equation*}
we will have the optimized value of $\omega$ as $\omega^*(\mathbf{x}) = \omega_0$. Specifically, when $\omega_0 = 1$, it will make the optimal-decay CBF-QP and CLF-CBF-QP as the original CBF-QP/CLF-CBF-QP. However, when the original CBF-QP and CLF-CBF-QP become infeasible, \textit{i.e.},
\begin{equation*}
    \forall \mathbf{u} \in \mathcal{U}_{adm}(\mathbf{x}), \ L_f h(\mathbf{x}) + L_g h(\mathbf{x}) \mathbf{u} < - \alpha(h(\mathbf{x})),
\end{equation*}
our optimal-decay form is still feasible and the optimal value of $\omega$ is as follows
\begin{equation*}
    \omega^*(\mathbf{x}) = \sup_{\mathbf{u} \in \mathcal{U}_{adm}} \dfrac{L_f h(\mathbf{x}) + L_g h(\mathbf{x}) \mathbf{u}}{-\alpha(h(\mathbf{x}))}.
\end{equation*}
To sum up, when $p_{\omega} = + \infty$ and $\omega_0 = 1$, the optimal-decay CBF-QP and optimal-decay CLF-CBF-QP are equivalent to original CBF-QP and CLF-CBF-QP, if the original ones are already feasible.
When the original constraint is infeasible, the optimization variable $\omega$ makes the proposed one still feasible.
\end{remark}

\section{Case Study: Adaptive Cruise Control}
\label{sec:example}
% Transition
Having presented our optimal-decay CBF-QP/CLF-CBF-QP formulation, we proceed to validate the proposed strategy using an adaptive cruise control (ACC) example, which has been commonly used to validate safety-critical control strategies \cite{ames2014control, xiao2019control}.

\subsection{Simulation Setup}
\label{subsec:result-feasibility}
Consider a point-mass model of an ego vehicle moving along a straight line to follow a lead vehicle. The dynamics are given as follows
\begin{equation}
    \dot{\mathbf{x}} = \begin{bmatrix} x_2 \\ -\dfrac{1}{m} F_r(\mathbf{x}) \\ v_l - x_2
    \end{bmatrix} + \begin{bmatrix} 0 \\ \dfrac{1}{m} \\ 0
    \end{bmatrix} \mathbf{u},
\end{equation}
where $(x_1, x_2)$ are the position and velocity ($x_2 = \dot{x}_1$) of the ego vehicle, $m$ is the mass of the vehicle, and $x_3$ is the distance between the vehicle and the lead vehicle traveling at a velocity of $v_l$. $F_r$ represent the aerodynamic drag and we use the empirical form given as below
\begin{equation}
    F_r = f_0 + f_1 x_2 + f_2 x_2^2,
\end{equation}
where $f_0$, $f_1$, and $f_2$ are empirical constants.

A CLF-CBF-QP controller is designed to solve this ACC problem. In order to regulate speed, we pick a control Lyapunov function as follows
\begin{equation}
    V = (x_2 - v_d)^2. \label{eq:acc-clf}
\end{equation}

A control barrier function is used to guarantee that the ego vehicle will not collide with the leading vehicle. We consider a control barrier function as 
\begin{equation}
    h = x_3 - 1.8x_2, \label{eq:acc-cbf}
\end{equation}
where the factor of 1.8 is a result of converting to SI units.
Besides speed regulation and safe distance maintenance, we also consider input constraints as follows
\begin{equation}
    c_d mg \leq \mathbf{u} \leq c_a mg,
\end{equation}
where $c_a$ and $c_d$ are coefficients for maximum and minimum wheel forces.

Using the control Lyapunov function \eqref{eq:acc-clf} and the control barrier function \eqref{eq:acc-cbf}, we design an original CLF-CBF-QP controller and an optimal-decay CLF-CBF-QP controller with formulation in \eqref{eq:clf-cbf-qp} and \eqref{eq:soft-clf-cbf-qp}, respectively. The system is simulated with the two controllers with a frequency at 100Hz, and the values of parameters in the simulations are shown in Tab. \ref{tab:parameter-values-acc-simulation}. Notice that in the simulations, we consider $\alpha(.)$ and $\gamma(.)$ in \eqref{eq:clf-cbf-qp} and \eqref{eq:soft-clf-cbf-qp} as linear functions with constant coefficients $\alpha$ and $\gamma$.

\subsection{Point-wise Feasibility}
To compare the performance between our optimal-decay CLF-CBF-QP and the original form, we simulate the system with different initial conditions and illustrate them together in Fig. \ref{fig:benchmark-different-initial-conditions}, where solid and dashed lines represents the results from our optimal-decay CLF-CBF-QP and the original form, respectively. For our optimal-decay form, we firstly pick the hyperparameters $\omega_0$ and $p_{\omega}$ as 1 and $10^{8}$, which corresponds to the choice we discussed in Sec. \ref{subsec:convergence-hyperparameter}.

Precisely, the ego vehicle is initialized at origin $x_1(0) = 0$ m, the initial distance between the ego car and the leading one is set as $x_3(0) = 100$ m and initial speed $x_2(0)$ ranges from 26 m/s to 32 m/s, shown in Fig. \ref{subfig:benchmark-feasibility-velocity}. We can see that CLF-CBF-QP with initial speed as 30 m/s and 32 m/s will become infeasible during the simulations. This infeasibility comes from the conflict between input constraint and original CBF constraint and we can see clearly that the CLF-CBF-QP fails with input saturation and our optimal-decay form handles the problem properly, shown in Fig. \ref{subfig:benchmark-feasibility-control-input}.
The simulation with initial speed as 26 m/s and 28 m/s does not encounter any infeasible issues where the CLF-CBF-QP and our optimal-decay form shared the same control performance. Moreover, we notice $\omega$ becomes not equal to one only when there is a conflict between the input constraint and the original CBF constraint, shown in Fig. \ref{subfig:benchmark-feasibility-omega}.
The evolution of control barrier function for two controllers are shown in Fig. \ref{subfig:benchmark-feasibility-cbf}.
To sum up, the optimal-decay CLF-CBF-QP is equivalent to the original one when $\omega_0 = 1$ if the original one is already feasible, while still being as a feasible problem even when the original CLF-CBF-QP is infeasible.

\begin{figure}
    \centering
    \begin{subfigure}[t]{0.49\linewidth}
        \centering
        \includegraphics[width=0.99\linewidth]{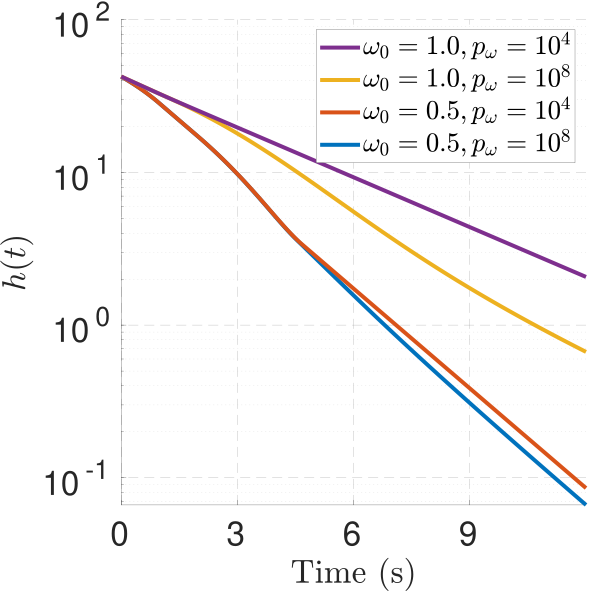}
        \caption{Different decay rates}
        \label{subfig:benchmark-hyperparameter-cbf}
    \end{subfigure}
    \begin{subfigure}[t]{0.49\linewidth}
        \centering
        \includegraphics[width=0.99\linewidth]{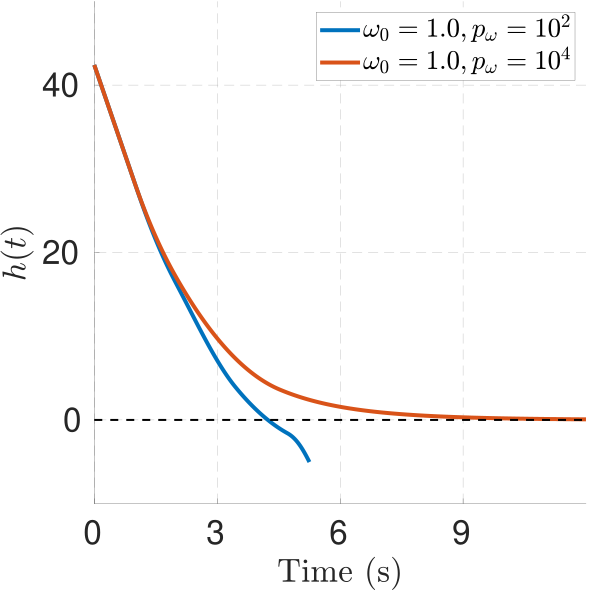}
        \caption{Different safety performances}
        \label{subfig:benchmark-safety-cbf}
    \end{subfigure}
    \caption{Illustration of the influence of hyperparameters $\omega_0$ and $p_{\omega}$ on CBF decay rates and system safety performance. The left plot is in log scale on the y-axis to show the difference in the decay rate.}
    \label{fig:benchmark-different-hyperparameters}
\end{figure}

\subsection{Convergence rate with hyperparameters}
\label{subsec:result-convergence}
As we have stated in Sec. \ref{subsec:persistent-feasibility-safety}, the system might become unsafe if it is initialized outside its invariant set.
To illustrate this property, we simulate the controller from an initial speed $x_2(0) = 32$ m/s and leading distance $x_3(0) = 100$ m with different values of hyperparameters $\omega_0$ and $p_{\omega}$.
We keep the same Lyapunov function and control barrier function with the same parameters in Table. \ref{tab:parameter-values-acc-simulation} and the simulation results are shown in Fig. \ref{fig:benchmark-different-hyperparameters}.
In Fig. \ref{subfig:benchmark-hyperparameter-cbf}, we can observe that smaller $\omega_0$ and larger $p_{\omega}$ would make the control barrier function decay slower, which verified exactly what was stated in Sec. \ref{subsec:convergence-hyperparameter}.

\subsection{Safety with hyperparameters}
\label{subsec:result-safety}
The simulations in Sec. \ref{subsec:result-feasibility} and \ref{subsec:result-convergence} that we have seen are all safe with set invariance on $\mathcal{C}$ defined in \eqref{eq:safe-set}.
However, for the same initial condition $x_2(0) = 32$ m/s, if we decrease $p_{\omega}$, the system could be unsafe after a while.  Specifically, we show that when $\omega_0 = 1, p_{\omega} = 10^2$, our system actually becomes unsafe after around 4s.
This happens since the initial condition no longer lies inside the invariant set for the optimal-decay CLF-CBF-QP with this configuration of hyperparameters, which was previously discussed in Sec. \ref{subsec:persistent-feasibility-safety}.
It must be noted that while a formulation guarantees point-wise feasibility for $h(\mathbf{x}) > 0$, this does not guarantee persistent feasibility for $h(\mathbf{x}) > 0$ in a sufficiently long-time trajectory.
This is seen in Fig. \ref{subfig:benchmark-safety-cbf}, where point-wise feasibility for all states with $h(\mathbf{x}) > 0$ still leads to the system unsafe after a while.

\begin{table} 
    \centering
    \begin{tabular}{|c c|c c|c c|}
    \hline
        \multicolumn{6}{|c|}{System Parameters} \\ \hline
        $m$ & 1650 kg & $v_l$ & 16 m/s & $v_d$ & 30 m/s\\
        $f_0$ & 0.1 N & $f_1$ & 5 Ns/m & $f_2$ & 0.25 $\text{N}\text{s}^2/\text{m}^2$ \\
        $c_d$ & -0.25 & $c_a$ & 0.25 & $g$ & 9.81 m/$\text{s}^2$ \\
        \hline
        \multicolumn{6}{|c|}{Controller Parameters} \\ \hline
        $\gamma$ & 1.0 & $p$ & 1.0 & $\alpha$ & 0.5 \\ \hline
    \end{tabular}
    \caption{Parameter values used in simulation.}
    \label{tab:parameter-values-acc-simulation}
\end{table}

% \begin{figure}
%     \centering
%     \begin{subfigure}[t]{0.49\linewidth}
%         \centering
%         \includegraphics[width=0.99\linewidth]{figures/benchmark-hyperparameter-cbf.png}
%         \caption{Different decay rates}
%         \label{subfig:benchmark-hyperparameter-cbf}
%     \end{subfigure}
%     \begin{subfigure}[t]{0.49\linewidth}
%         \centering
%         \includegraphics[width=0.99\linewidth]{figures/benchmark-safety-cbf.png}
%         \caption{Different safety performances}
%         \label{subfig:benchmark-safety-cbf}
%     \end{subfigure}
%     \caption{Illustration of the influence of hyperparameters $\omega_0$ and $p_{\omega}$ on CBF decay rates and system safety performance.}
%     \label{fig:benchmark-different-hyperparameters}
% \end{figure}

\section{Conclusion}
\label{sec:conclusion}
In this paper, we have analyzed the point-wise feasibility problem of safety critical control from the perspective of input-space and state-space.
Then, in order to deal with the conflict between CBF constraint and input constraint, we proposed an optimal-decay form for CBF-QP and CLF-CBF-QP to guarantee point-wise feasibility inside the safe set.
We numerically verified this control design in an ACC example.

\bibliographystyle{IEEEtran}
\balance
\bibliography{references}{}
\end{document}